\documentclass[11pt]{article}
\usepackage{amssymb}
\usepackage{amsthm}
\usepackage{amsmath}
\usepackage{geometry}
\usepackage{enumitem}
\usepackage{tikz}
\usepackage{url}
\usepackage{graphicx}
\usepackage{xcolor}
\usepackage{pgfplots}
\usepgfplotslibrary{fillbetween}
\usetikzlibrary{shapes.geometric}
\usetikzlibrary{arrows}

\pgfplotsset{ticks=none,compat=newest}

\numberwithin{equation}{section}

\newtheorem{thm}{Theorem}[section]
\newtheorem{prop}[thm]{Proposition}
\newtheorem{lem}[thm]{Lemma}
\newtheorem{cor}[thm]{Corollary}

\theoremstyle{definition}
\newtheorem{defn}[thm]{Definition}
\newtheorem{ass}[thm]{Assumption}

\theoremstyle{remark}
\newtheorem*{rem}{Remark}
\newtheorem*{exmp}{Example}

\newcommand{\diff}{\mathop{}\!\mathrm{d}}
\newcommand{\supp}{\mathop{}\!\mathrm{supp}}
\newcommand{\scp}[2]{\langle #1, #2\rangle}
\newcommand{\e}{\mathop{}\!\mathrm{e}}
\newcommand{\I}{\mathop{}\!\mathrm{i}}
\newcommand{\ran}{\mathop{}\!\mathrm{ran}}

\newcommand{\jap}[1]{\langle #1 \rangle}

\title{Mourre theory and spectral analysis of energy-momentum operators \\ in relativistic quantum field theory}
\author{Janik Kruse \\\\
	\small{Adam Mickiewicz University in Pozna\'n} \\
	\small{Faculty of Mathematics and Computer Science} \\  	
	\small{ul.~Uniwersytetu Pozna\'nskiego 4, 61-614 Pozna\'n, Poland} \\
	\small{E-mail: \tt{janik.kruse@amu.edu.pl}}}
\date{August 26, 2024}

\begin{document}

\maketitle

\begin{abstract}
	A central task of theoretical physics is to analyse spectral properties of quantum mechanical observables. In this endeavour, Mourre's conjugate operator method emerged as an effective tool in the spectral theory of Schrödinger operators. This paper introduces a novel class of examples from relativistic quantum field theory that are amenable to Mourre's method. By assuming Lorentz covariance and the spectrum condition, we derive a limiting absorption principle for the energy-momentum operators and provide new proofs of the absolute continuity of the energy-momentum spectra. Moreover, under the assumption of dilation covariance, we show that the spectrum of the relativistic mass operator is purely absolutely continuous in $(0,\infty)$.
\end{abstract}


\vspace*{0.1cm}


\section{Introduction}

Understanding spectral properties of quantum mechanical observables is of fundamental importance in theoretical physics. By analysing the spectrum, insights into the system's stability and long-time evolution can be obtained. Specifically, the spectrum of a self-adjoint operator decomposes into a pure point, absolutely continuous and singular continuous part, where the pure-point part corresponds to bound states and the absolutely continuous part to scattering states. Typically, the singular continuous spectrum, which evades a simple physical interpretation, is empty. One of the main objectives of this paper is to prove the absence of singular continuous spectrum for relativistic energy-momentum operators.

Various mathematical techniques have been developed to investigate the spectrum of a self-adjoint operator. One powerful method in this endeavour is Mourre's conjugate operator method, which is based on a strictly positive commutator estimate. The idea of analysing spectral properties of self-adjoint operators through commutator identities and estimates can be traced back to the pioneering works of Putnam \cite{putnam1967}, Kato \cite{kato1968}, and Lavine \cite{lavine1969}. However, it was Mourre \cite{mourre1981} who advanced this approach by introducing local commutator estimates. 

Mourre's method led to substantial progress in the spectral and scattering theory of Schrödinger operators. In his seminal paper, Mourre \cite{mourre1981} demonstrated the absence of singular continuous spectrum for 2- and 3-body Schrödinger operators. The applicability of the method to $N$-body systems was extended by Perry, Sigal, and Simon \cite{perry1981} and by Froese and Herbst \cite{froese1982}. Mourre's method also played a decisive role in establishing asymptotic completeness of $N$-body Schrödinger operators (see \cite{derezinski1997} for a textbook exposition). On an abstract level, the mathematical theory underlying the conjugate operator method was notably improved by Amrein, Boutet de Monvel, and Georgescu \cite{amrein1996}.

While originally developed for non-relativistic quantum mechanics, Mourre's method has also been extended to other areas. Worthy of note, it has been applied in non-relativistic quantum electrodynamics (QED). Among many works, we mention here that a Mourre estimate was proved by Dereziński and Gérard \cite{derezinski1999} for confined Pauli--Fierz Hamiltonians, by Fröhlich, Griesemer, and Schlein \cite{froehlich2004} for a Hamiltonian describing Compton scattering, by Fröhlich, Griesemer, and Sigal \cite{froehlich2008} for the standard model of non-relativistic QED, by Chen, Faupin, Fröhlich, and Sigal \cite{chen2011} for dressed electrons, and by Møller and Rasmussen \cite{moller2013} for the translation-invariant massive Nelson model. 

The application of Mourre's conjugate operator method in relativistic quantum field theory is more difficult due to the abstract nature of the Hamiltonian. Typically, the renormalised Hamiltonian is derived through a limiting procedure or is defined axiomatically as the generator of time translations. Although a Mourre estimate was established for the spatially cut-off $P(\varphi)_2$ Hamiltonian by Dereziński and Gérard \cite{derezinski2000} and Gérard and Panati \cite{gerard2008}, it remained an open problem whether a Mourre estimate can be proved for the Hamiltonian in the infinite-volume limit. Recently, the author \cite{kruse2024} applied Mourre's method within the axiomatic framework of Haag--Kastler quantum field theory to prove the existence of asymptotic observables. In this context, Mourre's method was implemented through scattering theory by comparing the abstract dynamics generated by the Hamiltonian to a more concrete free dynamics.

In this paper, we apply Mourre's conjugate operator method in the relativistic setting directly to the energy-momentum operators, which are defined axiomatically as the generators of spacetime translations. By assuming Lorentz covariance and the spectrum condition, we prove Mourre estimates for the energy-momentum operators $P=(P_0,\mathbf{P})$, using the generators of Lorentz boosts~$\mathbf{K}$ to construct conjugate operators. Our Mourre estimates yield the following limiting absorption principle:

\begin{thm} \label{thm:LAP}
	Let $U:\mathcal{P}\to\mathfrak{B}(\mathcal{H})$ be a strongly continuous unitary representation of the Poincaré group $\mathcal{P}=\mathcal{L}\ltimes\mathbb{R}^d$ on a Hilbert space $\mathcal{H}$, $P=(P_0,\mathbf{P})$ the generators of the translation subgroup~$U|_{\mathbb{R}^d}$, $E$ the joint spectral measure of $P$, and $\mathbf{K}$ the generators of Lorentz boosts. Assume that the energy-momentum operators $P$ obey the spectrum condition (see Assumption~\ref{ass:SpectrumCondition}). For $\kappa>0$, define the following Lorentz-invariant sets (see Figure~\ref{fig:Skappa}):
	\begin{align} \label{eq:Skappa}
		S_{\kappa} = \{ \Lambda_1(t_1)\dots\Lambda_s(t_s)(p_0,\mathbf{0}) \mid t_1,\dots,t_s\in\mathbb{R}, p_0\in [0,\kappa] \}^{-},
	\end{align}	
	where $\Lambda_j(t_j)$ are the Lorentz boosts in the spatial direction $j\in\{1,\dots,s=d-1\}$ and $\{\dots\}^{-}$ denotes the closure in $\mathbb{R}^d$. For all compact subsets $I_0\subset (\kappa,\infty)$ and $I_j \subset\mathbb{R}\backslash\{0\}$, for every $\nu>1/2$,
	\begin{align}
		\sup_{\lambda \in I_0, \mu > 0} \| E(S_\kappa) \jap{K_j}^{-\nu} (P_0-\lambda\mp\I\mu)^{-1} \jap{K_j}^{-\nu} E(S_\kappa) \| &< \infty, \\
		\sup_{\lambda \in I_j, \mu > 0} \| E(S_\kappa) \jap{K_j}^{-\nu} (P_j-\lambda\mp\I\mu)^{-1} \jap{K_j}^{-\nu} E(S_\kappa) \| &< \infty. \label{eq:LAP2}
	\end{align}
\end{thm}

\begin{figure}
	\begin{center}
		\begin{tikzpicture}		
			\fill [gray!20, domain=-5.292:5.292, variable=\x] (-6,6) -- plot ({\x}, {sqrt(\x*\x+8)}) -- (6,6) -- (0,0) -- cycle;
			
			\draw [color=gray, line width=0.5pt] (0,0) -- (6,6);
			\draw [color=gray, line width=0.5pt] (0,0) -- (-6,6);
			
			\draw [domain=-5.292:5.292, smooth, variable=\x, gray, line width=0.5pt] plot ({\x}, {sqrt(\x*\x+8)});
			
			\draw [color=black, ->, line width=0.5pt] (0,0) -- (0,6) node [above] {$p_0$};
			\draw [color=black, ->, line width=0.5pt] (-6,0) -- (6,0) node [right] {$\mathbf{p}$};
			
			\node [below=3pt] at (1.4,2.6) {$S_\kappa$};
			\node at (0,2.78){--};
			\node [right=3pt] at (0,2.8){$\kappa$};
		\end{tikzpicture}
	\end{center}
	\caption{The set $S_\kappa\subset V_+$ is invariant under Lorentz boosts.}
	\label{fig:Skappa}
\end{figure}
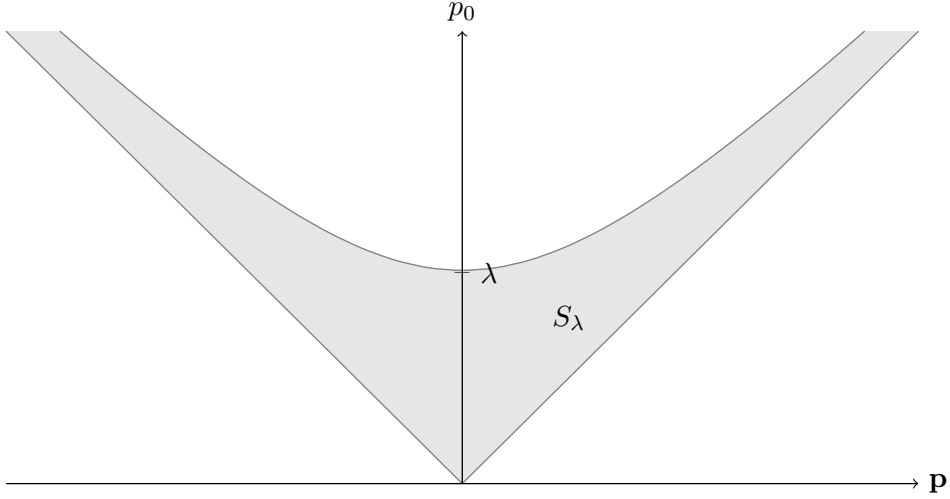

We illustrate our line of argument for the momentum operators $\mathbf{P}$. Formally, the momentum operator $P_j$ and the generator of Lorentz boosts $K_j$ in the spatial direction $j$ satisfy the following commutation relation:
\begin{align} \label{eq:CommutationRelationMomentumLorentzBoost}
	[P_j,\I K_j] = P_0.
\end{align}
By virtue of the spectrum condition, the energy $P_0$ is strictly positive on spectral subspaces of $P_j$ that are separated from 0. To make sense of the commutator and to apply the results of Mourre theory, it is necessary to demonstrate that $P_j$ lies in the regularity classes $C^k(K_j)$ (see Definition~\ref{defn:RegularityClasses}). However, this is not generally the case: because the inclusion $D(P_0) \subset D(P_j)$ of domains, which follows from the spectrum condition, can be proper, the form defined by the commutator $[P_j,\I K_j]$ on $D(P_j)\cap D(K_j)$ does not necessarily have a continuous extension to $D(P_j)$. To address this problem, we introduce the Lorentz-invariant sets $S_\kappa$, which cover the closed forward light cone $V_+$. Within the sets $S_\kappa$, the energy stays bounded relative to the momentum, ensuring that the operators $P_{0,\kappa} = E(S_\kappa)P_0$ and $P_{j,\kappa} = E(S_\kappa) P_j$ are bounded relative to each other. Moreover, the commutation relation \eqref{eq:CommutationRelationMomentumLorentzBoost} is preserved with $P_0$ and $P_j$ replaced by $P_{0,\kappa}$ and $P_{j,\kappa}$, respectively, because $E(S_\kappa)$ commutes with $K_j$. Additionally, we have $D(P_{0,\kappa}) = D(P_{j,\kappa})$ and $P_{j,\kappa} \in C^\infty(K_j)$ (see Proposition~\ref{prop:RegularityMomentumOperators}). The limiting absorption principle \eqref{eq:LAP2} can then be derived from the standard results of Mourre theory.

As a corollary of Theorem~\ref{thm:LAP}, we provide a new proof of the well-known theorem asserting that the energy-momentum spectra are purely absolutely continuous if translation-invariant vectors are removed \cite{maison1968}. More precisely, we prove the following proposition, where $\mathcal{H}_{\mathrm{ac}}(A)$ denotes the absolutely continuous spectral subspace of a commuting family of self-adjoint operators $A$ and $x\cdot y=x_0y_0-\mathbf{x}\cdot\mathbf{y}$, $x,y\in\mathbb{R}^d$, denotes the Minkowski scalar product.

\begin{prop} \label{prop:ACLorentzCovariance}
	Let $U$ be as in Theorem~\ref{thm:LAP} and $e\in\mathbb{R}^d\backslash\{0\}$ a spacetime vector such that $e\cdot e\neq 0$ if $d=2$. Assume that the energy-momentum operators $P=(P_0,\mathbf{P})$ obey the spectrum condition. If $Q_0$ denotes the projection onto the subspace of translation-invariant vectors, then
	\begin{align}\label{eq:ACLorentzCovariance}
		\mathcal{H}=Q_0\mathcal{H}\oplus \mathcal{H}_{\mathrm{ac}}(e\cdot P) = Q_0\mathcal{H}\oplus \mathcal{H}_{\mathrm{ac}}(\mathbf{P}).
	\end{align}
\end{prop}

The first identity in \eqref{eq:ACLorentzCovariance} was originally proved by Maison~\cite[Satz~2]{maison1968} for $d=4$ through the application of Wigner's theorem. The second identity (absolute continuity of the joint momentum spectrum) has been proved under differing assumptions in the realm of quantum field theory. Buchholz and Fredenhagen~\cite[Proposition~2.2]{buchholz1982} utilised the locality principle to establish this identity, and Bachmann, Dybalski, and Naaijkens \cite[Lemma~4.16]{bachmann2016} presented a simplified proof under the additional assumption of the existence of a vacuum vector.

An immediate consequence of Proposition~\ref{prop:ACLorentzCovariance} and the Riemann--Lebesgue lemma is the following clustering property. If $f,g\in\mathcal{H}$ are two arbitrary vectors, then
\begin{align}
	\lim_{t\to \pm\infty} \scp{f}{U(te) g} = \scp{f}{Q_0 g}.
\end{align}
In quantum field theory, the clustering property is well-known for space-like directions (i.e.~$e\cdot e<0$) due to the locality principle. It is somewhat unexpected that this property is also valid for light-like (i.e.~$e\cdot e=0$) and time-like directions (i.e.~$e\cdot e>0$) under the assumption of Lorentz covariance.

In two spacetime dimensions (i.e.~$d=2$), the clustering property may not hold for light-like directions. This stems from the potential presence of massless excitations, as elaborated upon in the remark subsequent to Proposition~\ref{prop:ACLightLike}.

\begin{prop} \label{prop:ACLorentzCovarianceLightLike}
	Let $U$ be as in Theorem~\ref{thm:LAP} and let $E_{P_0\pm P_1}$ denote the spectral measure of $P_0\pm P_1$. If $d=2$, then
	\begin{align}
		\mathcal{H} = E_{P_0\pm P_1}(\{0\})\mathcal{H} \oplus \mathcal{H}_{\mathrm{ac}}(P_0\pm P_1).
	\end{align}
\end{prop}

Another operator of interest in quantum field theory is the mass operator $M=\sqrt{P_0^2-|\mathbf{P}|^2}$. We prove the absence of singular continuous spectrum for the mass operator if $U$ satisfies dilation covariance. A unitary representation $U:\mathbb{R}^d\to\mathfrak{B}(\mathcal{H})$ of the translation group $\mathbb{R}^d$ is dilation-covariant if a self-adjoint operator $D$ exists such that, for all $t\in\mathbb{R},x\in\mathbb{R}^d$,
\begin{align}
	\e^{\I tD} U(x) \e^{-\I tD} = U(\e^{-t}x),
\end{align}
that is, the adjoint action of $\e^{\I tD}$ scales the translation $x$ by the factor $\e^{-t}$.

\begin{prop} \label{prop:ACDilationCovariance}
	Let $U:\mathbb{R}^d\to\mathfrak{B}(\mathcal{H})$ be a strongly continuous unitary representation of the translation group $\mathbb{R}^d$. Assume that the generators $P=(P_0,\mathbf{P})$ of $U$ obey the spectrum condition. If $U$ is dilation-covariant, then, for every compact subset $I\subset (0,\infty)$ and every $\nu>1/2$, 
	\begin{align}\label{eq:LAPMass}
		\sup_{\lambda\in I, \mu>0} \|\jap{D}^{-\nu} (M-\lambda\mp\I\mu)^{-1} \jap{D}^{-\nu} \| < \infty.
	\end{align}
	The spectrum of $M$ is purely absolutely continuous in $(0,\infty)$.
\end{prop}

Given that the spectrum of the mass operator $M\geq 0$ is contained in the interval $[0,\infty)$, the proposition implies that the pure point spectrum of $M$ must be empty or $\{0\}$, while the singular continuous spectrum is empty. Under the assumptions of Proposition~\ref{prop:ACLorentzCovariance}, the mass operator may have other eigenvalues than 0. The assumption of dilation covariance excludes additional mass shells within the energy-momentum spectrum.

The structure of the paper is as follows. In Section~\ref{sec:LocallySmoothOperators}, we define locally smooth operators for a commuting family of self-adjoint operators, and we relate the existence of a locally smooth operator to the regularity of the joint spectrum. In Section~\ref{sec:Mourre}, we review essential results of Mourre's conjugate operator method. We explain how to obtain locally smooth operators based on a strictly positive commutator estimate. In Section~\ref{sec:SpectralPropertiesEMOperators}, we present the proofs of our results.

\section{Locally smooth operators} \label{sec:LocallySmoothOperators}

Let $U:\mathbb{R}^n\to\mathfrak{B}(\mathcal{H})$ be a strongly continuous unitary representation of $\mathbb{R}^n$ on a Hilbert space $\mathcal{H}$ with self-adjoint generators $H=(H_1,\dots,H_n)$. We denote by $D(H)$ the intersection of the domains of $H_1,\dots,H_n$, by $E$ the spectral measure of the family $H$, and by $\sigma(H)\subset\mathbb{R}^n$ its joint spectrum. We consider $D(H)$ as a Banach space equipped with the graph norm.

\begin{defn}\label{defn:LocallySmooth}
	A continuous operator $T:D(H)\to\mathcal{H}$ is \textbf{locally $H$-smooth} on an open set $J\subset\mathbb{R}^n$ if, for every compact subset $K\subset J$, a constant $C_{K}$ exists such that, for all $f\in\mathcal{H}$,
	\begin{align}\label{eq:HsmoothOperator}
		\int_{\mathbb{R}^{n}} \|TU(x)E(K)f\|^2 \diff x \leq C_{K} \|f\|^2.
	\end{align}
\end{defn}

In \cite[Definition~A.9]{kruse2024}, we defined locally $H$-smooth operators by demanding that all subfamilies of $H$ are locally $H$-smooth in the sense of the above definition. The advantage of this stronger definition is that it allows locally $H$-smooth operators to be equivalently characterised by resolvent estimates which resemble a limiting absorption principle (see \cite[Theorem~A.12]{kruse2024} and Proposition~\ref{prop:OptimalConstant1} below for the case $n=1$). In this paper, we prefer the weaker formulation of Definition~\ref{defn:LocallySmooth} because it is sufficient for analysing the regularity of the joint spectrum $\sigma(H)$. 

The following proposition is a straightforward generalisation of \cite[Theorem~XIII.23]{reed1978}.

\begin{prop} \label{prop:LocallySmoothACSpectrum}
	If $T$ is locally $H$-smooth on $J$, then $E(J)\overline{\ran(T^*)} \subset E(J)\mathcal{H}_{\mathrm{ac}}(H)$. If, additionally, $\ker(T)=\{0\}$, then the joint spectrum of $H$ is purely absolutely continuous in $J$.
\end{prop}

\begin{proof}
	Let $f\in E(J)\overline{\ran(T^*)}$. Due to the inner regularity of the spectral measure $E$, we can approximate $f$ by a sequence $\{f_j\}_{j\in\mathbb{N}}$ such that, for every $j\in\mathbb{N}$, $f_j\in E(K_j)\ran(T^*)$ for a compact subset $K_j\subset J$. If, for every $j\in\mathbb{N}$, $f_j \in E(J)\mathcal{H}_{\mathrm{ac}}(H)$, then also $f\in E(J)\mathcal{H}_{\mathrm{ac}}(H)$ because $E(J)\mathcal{H}_{\mathrm{ac}}(H)$ is a closed subspace of $\mathcal{H}$. We write $f_j=E(K_j)T^*g_j$ with $g_j\in D(T^*)$. Let $\mu_{f_j}$ be the spectral measure associated with $f_j$, and let $F_j$ be its inverse Fourier transform:
	\begin{align}
		F_j(x) = \frac{1}{(2\pi)^n} \int_{\mathbb{R}^n} \e^{\I p\cdot x} \diff\mu_{f_j}(p) = \frac{1}{(2\pi)^n} \scp{E(K_j)T^*g_j}{U(x)f_j}.
	\end{align}
	The function $F_j$ is square-integrable because $|F_j(x)|\leq (2\pi)^{-n}\|g_j\| \|TU(x)E(K_j)f_j\|$ and $T$ is locally $H$-smooth on $J$ by assumption. Thus, $\diff\mu_{f_j}(p)=\hat{F}_j(p)\diff p$ is an absolutely continuous measure.
	
	To prove the second statement, we apply the decomposition $\mathcal{H}=\ker(T)\oplus \overline{\ran(T^*)}$. If $\ker(T)=\{0\}$, then $E(J)\mathcal{H}=E(J)\overline{\ran(T^*)}\subset E(J)\mathcal{H}_{\mathrm{ac}}(H)$. We conclude $E(J)\mathcal{H}=E(J)\mathcal{H}_{\mathrm{ac}}(H)$ because $\mathcal{H}_{\mathrm{ac}}(H)\subset \mathcal{H}$, that is, the joint spectrum of $H$ is purely absolutely continuous in $J$.
\end{proof}

It is possible to determine the constant $C_K$ for the bound \eqref{eq:HsmoothOperator}, which is useful to verify the $H$-smoothness of an operator $T$. The following proposition is proved as in \cite[Proposition~A.10]{kruse2024}.

\begin{prop}\label{prop:OptimalConstant}
	A continuous operator $T:D(H)\to\mathcal{H}$ is locally $H$-smooth on $J$ if and only if, for all compact subsets $K\subset J$,
	\begin{align}
		C_K^0 = 2^n \sup_{\lambda\in\mathbb{R}^n, \mu\in(0,1)^n} \mu_1\dots\mu_n \|TE(K)R(\lambda+\I\mu)\|^2 < \infty,
	\end{align}
	where
	\begin{align}
		R(\lambda+\I\mu) = \prod_{j=1}^n (H_j-\lambda_j-\I\mu_j)^{-1}.
	\end{align}
\end{prop}

If $n=1$, we can formulate a necessary and sufficient condition for $C_K^0<\infty$ in terms of a limiting absorption principle (see Theorem~\ref{thm:LimitingAbsorptionPrinciple}).

\begin{prop}\label{prop:OptimalConstant1}
	Let $n=1$. For every compact subset $K\subset\mathbb{R}$, $C_K^0<\infty$ if and only if
	\begin{align}\label{eq:UniformResolventEstimate}
		\sup_{\lambda\in K,\mu\in(0,1)} \|T\Im R(\lambda+\I\mu)T^*\| < \infty.
	\end{align}
\end{prop}

The proof of Proposition~\ref{prop:OptimalConstant1} is contained in the proof of \cite[Proposition~7.1.1]{amrein1996}; see \cite[Theorem~A.12]{kruse2024} for a generalisation to $n\geq 2$. In the next section, we apply Proposition~\ref{prop:OptimalConstant} and Proposition~\ref{prop:OptimalConstant1} to construct locally smooth operators (see Corollary~\ref{cor:AbsolutelyContinuousSpectrum}).

\section{Mourre's conjugate operator method} \label{sec:Mourre}

Mourre's conjugate operator method is a mathematical tool to analyse the spectrum of a self-adjoint operator based on a positive commutator estimate. In this section, we outline the framework and state the most important results of the method, in particular the limiting absorption principle. At the end of this section, we apply Mourre's method to construct locally smooth operators. 

We introduce the following regularity classes, which are relevant for defining commutators.

\begin{defn} \label{defn:RegularityClasses}
	Let $A$ be a self-adjoint operator on $\mathcal{H}$ and let $k\in\mathbb{N}\cup\{\infty\}$. We denote by $C^k(A)$ the space of all self-adjoint operators $H$ such that $t\mapsto \e^{\I t A}(H+\I)^{-1}\e^{-\I t A}$ is a $C^k$-map in the strong operator topology. We denote by $C_\mathrm{u}^k(A)$ the subspace of operators $H$ for which the same map is $C^k$ in norm.
\end{defn}

We collect some properties of the class $C^1(A)$ from \cite[Theorem~6.2.10]{amrein1996}. If $H\in C^1(A)$, then $D(A)\cap D(H)$ is dense in $D(H)$ (equipped with the graph topology). Moreover, the sesquilinear form defined by the commutator $HA-AH$ on $D(A)\cap D(H)$ extends to a continuous sesquilinear form on $D(H)$. We denote the extended sesquilinear form by $[H,A]$. It is possible to characterise the class $C^1(A)$ in terms of the commutator $[H,A]$. 

\begin{prop}[{\cite[Theorem~6.2.10~(a)]{amrein1996}}] \label{prop:C1CommutatorCharacterisation}
	A self-adjoint operator $H$ belongs to the class $C^1(A)$ if and only if there exists $z\in \rho(H)$ such that $\{f\in D(A)\mid (H-z)^{-1}f\in D(A) \text{ and } (H-\overline{z})^{-1}f\in D(A)\}$ is a core for $A$ and, for all $f\in D(A)\cap D(H)$,
	\begin{align}
			|\scp{Hf}{Af}-\scp{Af}{Hf}|\leq c(\|Hf\|^2+\|f\|^2).
		\end{align}
\end{prop}

Having formalised the commutator of two self-adjoint unbounded operators, we are able to define strictly positive commutator estimates.

\begin{defn} \label{defn:MourreEstimate}
	A self-adjoint operator $H$ obeys a \textbf{Mourre estimate} on an open and bounded subset $J\subset\mathbb{R}$ if a self-adjoint operator $A$ (\textbf{conjugate operator}) exists such that $H\in C^1(A)$ and, for an $a>0$,
	\begin{align}
		E(J) [H, \I A] E(J) \geq a E(J),
	\end{align}
	where $E$ is the spectral measure of $H$.
\end{defn}

\begin{exmp}
	Let $H$ be $A$-homogeneous, that is, for every $t,x\in\mathbb{R}$, 
	\begin{align}\label{eq:AHomogeneous}
		\e^{\I t A}\e^{\I xH}\e^{-\I tA} = \e^{\I \e^{-t}x H}.
	\end{align}
	It is not difficult to show that \eqref{eq:AHomogeneous} implies $H\in C^\infty(A)$ and $[H,\I A] = H$. Moreover, for every open and bounded interval $J=(a,b)\subset(0,\infty)$, $0<a<b<\infty$,
	\begin{align}
		E(J)[H,\I A] E(J) = HE(J) \geq a E(J),
	\end{align}
	that is, $H$ obeys a Mourre estimate on $J$ with conjugate operator $A$. Similarly, for every open and bounded interval $I=(a,b)\subset(-\infty,0)$, $-\infty<a<b<0$,
	\begin{align}
		E(I)[H,-\I A] E(I) = -HE(I) \geq -b E(I),
	\end{align}
	that is, $H$ obeys a Mourre estimate on $I$ with conjugate operator $-A$.
\end{exmp}

A consequence of the Mourre estimate on $J$ is that $H$ has no eigenvalues in $J$. This fact emerges as a direct corollary of the virial theorem \cite[Proposition~7.2.10]{amrein1996}.

\begin{thm}[Virial theorem] \label{thm:VirialTheorem}
	Let $H\in C^1(A)$. If $f$ is an eigenvector of $H$, then $\scp{f}{[H,A]f}=0$.
\end{thm}

Next, we state the limiting absorption principle. To formulate the limiting absorption principle under optimal conditions, we introduce new regularity classes. Let $p\in[1,\infty]$ and $k\in\mathbb{N}$. For $0<s<k$, we denote by $\mathcal{C}^{s,p}(A)$ the real interpolation space \cite[(5.2.6)]{amrein1996}
\begin{align}
	\mathcal{C}^{s,p}(A) = (C^k_\mathrm{u}(A),C^0_{\mathrm{u}}(A))_{\theta,p}, \ \theta=1-s/k.
\end{align}
Moreover, we write $H\in C^{k+0}(A)$ \cite[p.~204]{amrein1996} if $H\in C^k(A)$ and if the operator-valued function $S(t)=\e^{\I tA}\mathrm{ad}_A^k((H+\I)^{-1})\e^{-\I tA}$ is Dini continuous in norm, that is, 
\begin{align}
	\int_{|t|<1} \|S(t)-S(0)\| \frac{\diff t}{t} < \infty.
\end{align}
We mention the following chain of inclusions \cite[(5.2.14)]{amrein1996}:
\begin{align}\label{eq:Inclusions}
	C^{k+1}(A) \subset C^{k+0}(A) \subset \mathcal{C}^{k,1}(A) \subset C^k_\mathrm{u}(A), \ k\in\mathbb{N}.
\end{align}
Note that $C^{k+1}(A)\subset C^{k+0}(A)$ as a consequence of the mean value theorem and the uniform boundedness principle.

We say that $H$ has a spectral gap if $\sigma(H)\neq \mathbb{R}$ and we write $\jap{A}$ for $\sqrt{A^2+1}$.

\begin{thm}[Limiting absorption principle] \label{thm:LimitingAbsorptionPrinciple}
	Let $H\in \mathcal{C}^{1,1}(A)$ if $H$ has a spectral gap and $H\in C^{1+0}(A)$ if $\sigma(H)=\mathbb{R}$. If $H$ obeys a Mourre estimate on $J$, then, for every compact subset $K\subset J$ and every $\nu>1/2$,
	\begin{align}\label{eq:LAP1}
		\sup_{\lambda\in K, \mu>0} \| \jap{A}^{-\nu} (H-\lambda\mp\I\mu)^{-1} \jap{A}^{-\nu} \| < \infty.
	\end{align}
\end{thm}

\begin{rem}
	The assumption $H\in \mathcal{C}^{1,1}(A)$ is optimal on the Besov scale in the sense that counterexamples of the limiting absorption principle for operators possessing less regularity exist \cite[Section~7.B]{amrein1996}. A proof of the limiting absorption principle under the assumption $H\in \mathcal{C}^{1,1}(A)$ has been accomplished in the case that $H$ has a spectral gap \cite[Theorem~7.4.1]{amrein1996}. Sahbani \cite{sahbani1997} proved the limiting absorption principle under the more restrictive assumption $H\in C^{1+0}(A)$. It is an open problem to prove the limiting absorption principle for $H\in \mathcal{C}^{1,1}(A)$ if $\sigma(H)=\mathbb{R}$.
\end{rem}

According to Proposition~\ref{prop:LocallySmoothACSpectrum}, $H$ has no singular spectrum in $J$ if an injective locally $H$-smooth operator on $J$ exists. If each $H_i$ obeys a Mourre estimate with conjugate operator $A_i$ such that, for $i\neq j$, $A_i$ commutes strongly with $H_j$, then, for every $\nu>1/2$, $\jap{A_1}^{-\nu}\dots\jap{A_n}^{-\nu}$ is an injective operator that is locally $H$-smooth. This is a corollary of the limiting absorption principle and the results from Section~\ref{sec:LocallySmoothOperators}.

\begin{cor} \label{cor:AbsolutelyContinuousSpectrum}
	Consider a commuting family of self-adjoint operators $H=(H_1,\dots,H_n)$ and a corresponding family $A=(A_1,\dots,A_n)$ of self-adjoint operators such that, for every $i\in\{1,\dots,n\}$, $H_i\in \mathcal{C}^{1,1}(A_i)$ if $H_i$ has a spectral gap and $H_i\in C^{1+0}(A_i)$ if $\sigma(H_i)=\mathbb{R}$. If $H_i$ obeys a Mourre estimate on $J_i$ with conjugate operator $A_i$ and, for $i\neq j$, $A_i$ commutes strongly with $H_j$, then the joint spectrum of $H$ is purely absolutely continuous in the product region $J=J_1\times\dots\times J_n$.
\end{cor}
	
\begin{proof}
	We demonstrate that the injective operator $T=\jap{A_1}^{-\nu}\dots\jap{A_n}^{-\nu}$ is locally $H$-smooth on~$J$ for every $\nu>1/2$. If this claim is proved, the statement follows from Proposition~\ref{prop:LocallySmoothACSpectrum}. We show that the constant $C_K^0$ from Proposition~\ref{prop:OptimalConstant} is finite for every compact subset $K\subset J$. It suffices to consider compact hyperrectangles $K=K_1\times \dots \times K_n$, where $K_i\subset J_i$ are compact intervals, because every compact subset $K\subset J$ can be covered by finitely many compact hyperrectangles in~$J$. If $E_i$ denotes the spectral measure of $H_i$, then
	\begin{align}
		E(K) = E_1(K_1)\dots E_n(K_n).
	\end{align}
	By assumption, $A_i$ commutes strongly with $H_j$ if $i\neq j$; hence,
	\begin{align}
		C_K^0 &= 2^n \sup_{\lambda\in \mathbb{R}^n, \mu\in(0,1)^n} \mu_1\dots\mu_n \|\jap{A_1}^{-\nu}\dots\jap{A_n}^{-\nu}E(K)R(\lambda+\I\mu)\|^2 \notag \\
		&\leq \prod_{i=1}^{n}2 \sup_{\lambda_i\in \mathbb{R},\mu_i\in(0,1)} \mu_i\|\jap{A_i}^{-\nu} E_i(K_i) R_i(\lambda_i+\I \mu_i)\|^2.
	\end{align}
	Each factor on the r.h.s.~is finite due to Proposition~\ref{prop:OptimalConstant1} and the limiting absorption principle (Theorem~\ref{thm:LimitingAbsorptionPrinciple}).
\end{proof}

\section{Spectral analysis of relativistic energy-momentum operators}
\label{sec:SpectralPropertiesEMOperators}

In this section, we establish Mourre estimates for the energy-momentum operators $P=(P_0,\mathbf{P})$ and verify the limiting absorption principle of Theorem~\ref{thm:LAP} (see Theorem~\ref{thm:LAPMomentum} and Theorem~\ref{thm:LAPEnergy} below). Moreover, we analyse the spectra of the energy-momentum operators: in Proposition~\ref{prop:ACMomentum} for the momentum operators $\mathbf{P}$, in Proposition~\ref{prop:ACEnergy} for the energy operator $P_0$, and in Proposition~\ref{prop:ACLightLike} for the operators $e\cdot P$, where $e$ is a light-like vector. The relativistic mass operator $M$ and its spectrum is studied in Section~\ref{ssec:DilationCovariantRepresentations}.

\subsection{Representations of the Poincaré group}

The Poincaré group is the semi-direct product $\mathcal{P}=\mathcal{L}\ltimes\mathbb{R}^d$ of the Lorentz group $\mathcal{L}=O(d-1,1)$ and the translation group $\mathbb{R}^d$. Its multiplication law is defined as follows:
\begin{align}
	(\lambda_1,a_1)\cdot (\lambda_2,a_2) = (\lambda_1 \lambda_2, a_1 + \lambda_1 a_2), \ \lambda_1, \lambda_2\in\mathcal{L}, a_1,a_2\in\mathbb{R}^d.
\end{align}
Let $U:\mathcal{P}\to\mathfrak{B}(\mathcal{H})$ be a strongly continuous unitary representation of the Poincaré group on a Hilbert space $\mathcal{H}$. The generators of the translation subgroup $U|_{\mathbb{R}^d}$ are the energy-momentum operators $P=(P_0,\mathbf{P})$, such that, for $x\in\mathbb{R}^d$, $U(1,x)=\e^{\I x\cdot P}$, where $x\cdot P = x_0 P_0 - \mathbf{x}\cdot\mathbf{P}$. Let $E$ be the joint spectral measure of $P$. A vector $f\in\mathcal{H}$ is translation-invariant if $U(1,x)f=f$ for all $x\in\mathbb{R}^d$. We denote by $Q_0=E(\{0\})$ the projection onto the subspace of translation-invariant vectors.

Let $\Lambda_1(t),\dots,\Lambda_s(t)$ be the Lorentz boosts in the space directions $1,\dots,s=d-1$, for example,
\begin{align} \label{eq:LorentzBoosts}
	\Lambda_1(t) = \begin{pmatrix}
		\cosh(t) & \sinh(t) & 0 & \cdots & 0 \\ 
		\sinh(t) & \cosh(t) & 0 & \cdots & 0 \\ 
		0 & 0 & 1 &  & 0 \\ 
		\vdots & \vdots & & \ddots & \\ 
		0 & 0 & 0 & & 1
	\end{pmatrix},
\end{align}
and let $K_1,\dots,K_s$ be the self-adjoint generators of the Lorentz boosts, that is, $U(\Lambda_j(t),0)=\e^{\I tK_j}$, $j\in\{1,\dots,s\}$. We mention the following identities, which are consequences of the multiplication law ($t\in\mathbb{R}$, $x\in\mathbb{R}^d$):
\begin{align}
	\e^{\I t K_j}\e^{-\I x_0 P_0}\e^{-\I t K_j} &= \e^{-\I x_0 (\cosh(t)P_0-\sinh(t)P_j)}, \label{eq:CommutationEnergyLorentz} \\
	\e^{\I t K_j}\e^{-\I x_j P_j}\e^{-\I t K_j} &= \e^{-\I x_j (\cosh(t)P_j-\sinh(t)P_0)}. \label{eq:CommutationMomentumLorentz}
\end{align}
Formally, by differentiating in $x=0$ and $t=0$, these identities are equivalent to the commutation relations $[P_0,\I K_j]=P_j$ and $[P_j,\I K_j]=P_0$.

Our analysis below crucially depends on the spectrum condition. The spectrum condition states that the energy-momentum spectrum is contained within the forward light cone, which is the largest Lorentz-invariant set where the energy is nonnegative. This condition is a quantum field theory axiom in the frameworks of Wightman and Haag–Kastler. It will be assumed for the remainder of this section.

\begin{ass}[Spectrum condition] \label{ass:SpectrumCondition}
	The joint spectrum of the energy-momentum operators~$P$ is a subset of the closed forward light cone $V_+=\{p=(p_0,\mathbf{p})\in\mathbb{R}^d \mid p_0 \geq |\mathbf{p}| \}$ (i.e. $\sigma(P)\subset V_+$).
\end{ass}

From the spectrum condition, it follows that the momentum operators are relatively bounded relative to the energy operator $P_0$, that is, $|\mathbf{P}|\leq P_0$. This implies the inclusion $D(P_0) \subset D(P_j)$ of domains, which can be proper. As explained in the introduction, this is problematic for defining the commutators $[P_0, \I K_j]$ and $[P_j,\I K_j]$ in a way suitable for Mourre theory. Specifically, $P_0$ and $P_j$ are not necessarily elements of the regularity classes $C^k(K_j)$. 

The Lorentz-invariant sets $S_\kappa$, defined in \eqref{eq:Skappa}, cover the light cone $V_+$ (i.e.~$\bigcup_{\kappa>0}S_\kappa = V_+$). These sets are constructed so that the energy within $S_\kappa$ remains bounded relative to the momentum. Consequently, the operators $P_{0,\kappa} = E(S_\kappa)P_0$ and $\mathbf{P}_\kappa=E(S_\kappa)\mathbf{P}$ are bounded relative to each other:
\begin{align}
	|\mathbf{P}_\kappa| \leq P_{0,\kappa} \leq C_\kappa (1+|\mathbf{P}_\kappa|).
\end{align}
Moreover, the subspaces $E(S_\kappa)\mathcal{H}$ cover the Hilbert space $\mathcal{H}$, and it holds that
\begin{align}
	\mathcal{H}_{\mathrm{ac}}(\mathbf{P}) = \bigcup_{\kappa>0} \mathcal{H}_{\mathrm{ac}}(\mathbf{P}_\kappa), \ \mathcal{H}_{\mathrm{ac}}(P_0) = \bigcup_{\kappa>0} \mathcal{H}_{\mathrm{ac}}(P_{0,\kappa}).
\end{align}
In the following three subsections, we prove Mourre estimates and absence of singular continuous spectrum for the momentum operators $\mathbf{P}$, the energy operator $P_0$, and the light-cone operators $e\cdot P$, where $e$ is a light-like vector.

\subsubsection{Momentum operators}

The following proposition, which proves that $P_{0,\kappa}$ and $P_{j,\kappa}$ are elements of the regularity class $C^\infty(K_j)$, is essential for establishing the Mourre estimate and applying the results of Section~\ref{sec:Mourre}.

\begin{prop}\label{prop:RegularityMomentumOperators}
	For every $\kappa>0$ and $j\in\{1,\dots,s\}$, $P_{0,\kappa}\in C^\infty(K_j)$, $P_{j,\kappa}\in C^\infty(K_j)$, and 
	\begin{align}
		[P_{0,\kappa},\I K_j] = P_{j,\kappa}, \\
		[P_{j,\kappa},\I K_j]=P_{0,\kappa}.
	\end{align}
\end{prop}

\begin{proof}
	The spectral projection $E(S_\kappa)$ commutes with $K_j$ because $S_\kappa$ is a Lorentz-invariant set. Thus, we can replace $P_0$ and $P_j$ in \eqref{eq:CommutationMomentumLorentz} with $P_{0,\kappa}$ and $P_{j,\kappa}$: 
	\begin{align}\label{eq:CommutationMomentumLorentzLambda}
		\e^{\I t K_j}\e^{-\I x_j P_{j,\kappa}}\e^{-\I t K_j} &= \e^{-\I x_j(\cosh(t)P_{j,\kappa}-\sinh(t)P_{0,\kappa})}.
	\end{align}
	From \eqref{eq:CommutationMomentumLorentzLambda}, we obtain the following identity:
	\begin{align}
		\e^{\I tK_j} (P_{j,\kappa}+\I)^{-1} \e^{-\I tK_j} &= (\cosh(t)P_{j,\kappa}-\sinh(t)P_{0,\kappa}+\I)^{-1}.
	\end{align}
	The r.h.s.~is a smooth function in $t$ in the strong operator topology because $P_{0,\kappa}$ and $P_{j,\kappa}$ are bounded relative to each other; hence, $P_{j,\kappa}\in C^\infty(K_j)$. Moreover, by differentiating \eqref{eq:CommutationMomentumLorentzLambda} in $t=0$ and $x_j=0$, we obtain, for every $f,g\in D(P_{j,\kappa})\cap D(K_j)$,
	\begin{align}\label{eq:QuadraticFormMomentumLorentz}
		\scp{P_{j,\kappa}f}{K_jg} - \scp{K_jf}{P_{j,\kappa}g} = -\I\scp{f}{P_{0,\kappa}g}.
	\end{align}
	Because $P_{j,\kappa}\in C^1(K_j)$, the sesquilinear form \eqref{eq:QuadraticFormMomentumLorentz} has a unique extension to $D(P_{j,\kappa})=D(P_{0,\kappa})$, yielding the commutator identity $[P_{j,\kappa},\I K_j] = P_{0,\kappa}$.	
	The proof of $P_{0,\kappa}\in C^\infty(K_j)$ and $[P_{0,\kappa},\I K_j] = P_{j,\kappa}$ is analogous.
\end{proof}

\begin{thm} \label{thm:LAPMomentum}
	Under the assumptions of Theorem~\ref{thm:LAP}, the following limiting absorption principle holds for every $\kappa>0$, every compact subset $I_j\subset\mathbb{R}\backslash\{0\}$, and every $\nu>1/2$:
	\begin{align} \label{eq:LAPEnergy}
		\sup_{\lambda \in I_j, \mu > 0} \| E(S_\kappa) \jap{K_j}^{-\nu} (P_j-\lambda\mp\I\mu)^{-1} \jap{K_j}^{-\nu} E(S_\kappa) \| < \infty.
	\end{align}
\end{thm}

\begin{proof}
	Let $a>0$ and let $J$ be an open and bounded subset of $(-\infty,-a]\cup[a,\infty)$. From Proposition~\ref{prop:RegularityMomentumOperators} and the spectrum condition ($P_{0,\kappa} \geq |P_{j,\kappa}|$), we obtain the following Mourre estimate:
	\begin{align}
		E_{j,\kappa}(J) [P_{j,\kappa},\I K_j] E_{j,\kappa}(J) = P_{0,\kappa} E_{j,\kappa}(J) \geq a E_{j,\kappa}(J),
	\end{align}
	where $E_{j,\kappa}$ denotes the spectral measure of $P_{j,\kappa}$. Moreover, $P_{j,\kappa}\in C^{1+0}(K_j)$ because $C^\infty(K_j) \subset C^{1+0}(K_j)$ by \eqref{eq:Inclusions}. Thus, the limiting absorption principle \eqref{eq:LAPEnergy} follows from Theorem~\ref{thm:LimitingAbsorptionPrinciple} and the fact that $E(S_\kappa)$ commutes with $K_j$.
\end{proof}

\begin{prop} \label{prop:ACMomentum}
	Under the assumptions of Proposition~\ref{prop:ACLorentzCovariance}, $\mathcal{H} = Q_0\mathcal{H} \oplus \mathcal{H}_{\mathrm{ac}}(\mathbf{P})$, and, for every space-like vector $e\in\mathbb{R}^d$, $\mathcal{H} = Q_0\mathcal{H} \oplus \mathcal{H}_{\mathrm{ac}}(e\cdot P)$.
\end{prop}

\begin{proof}
	In the proof of Theorem~\ref{thm:LAPMomentum}, we demonstrated that $P_{j,\kappa}$ obeys a Mourre estimate on every open and bounded subset of $(-\infty,-a]\cup[a,\infty)$, $a>0$, with conjugate operator $K_j$. Moreover, $K_i$ commutes strongly with $P_{j,\kappa}$ if $i\neq j$. Because $a>0$ can be arbitrary small, it follows from Corollary~\ref{cor:AbsolutelyContinuousSpectrum} that $E(\{p \in S_\kappa \mid \forall j\in \{1,\dots,s\} \colon p_j\neq0 \})\mathcal{H}\subset \mathcal{H}_{\mathrm{ac}}(\mathbf{P}_\kappa)$. By taking the union over $\kappa>0$, we obtain $E(\{p \in V_+ \mid \forall j\in \{1,\dots,s\} \colon p_j\neq 0 \})\mathcal{H}\subset \mathcal{H}_{\mathrm{ac}}(\mathbf{P})$, and, according to Lemma~\ref{lem:TranslationInvariantVectors} below, $E(\{p \in V_+ \mid \forall j\in \{1,\dots,s\} \colon p_j\neq 0 \})\mathcal{H}=E(V_+\backslash\{0\})\mathcal{H}$. Thus, the first statement of the proposition follows from the decomposition $\mathcal{H}=Q_0\mathcal{H} \oplus E(V_+\backslash\{0\})\mathcal{H}$. The second statement can be proved by a similar argument or can be derived from the first statement.
\end{proof}

\begin{lem}\label{lem:TranslationInvariantVectors}
	Let $f\in\mathcal{H}$. If $P_\mu f=0$ for one $\mu\in\{0,\dots,s\}$, then $P_\mu f=0$ for all $\mu\in\{0,\dots,s\}$, that is, if a vector is translation-invariant in one spacetime direction, then it is translation-invariant in all spacetime directions.
\end{lem}

\begin{proof}
	If $P_0f=0$, then $P_jf=0$ for all $j\in\{1,\dots,s\}$ due to the spectrum condition. If $P_jf=0$ for one $j\in\{1,\dots,s\}$, then also $P_{j,\kappa}f=0$ for all $\kappa>0$. From the commutation relation $[P_{j,\kappa},\I K_j]=P_{0,\kappa}$ and the virial theorem (Theorem~\ref{thm:VirialTheorem}), it follows that
	\begin{align}
		\|\sqrt{P_{0,\kappa}}f\|^2 = \scp{f}{[P_{j,\kappa},\I K_j]f}=0;
	\end{align}
	hence, $P_{0,\kappa} f=0$ for all $\kappa>0$, that is, $P_0 f=0$. 
\end{proof}

\subsubsection{Energy operator}

Constructing a conjugate operator for the energy operator $P_0$ is more difficult. We cannot choose the generator $K_j$ of a Lorentz boost because $[P_{0,\kappa},K_j] = P_{j,\kappa}$, and $P_{j,\kappa}$ has no definite sign on any spectral subspace of $P_{0,\kappa}$. The conjugate operator $\overline{A_\kappa}$, which we construct below, is adapted from \cite[Lemma~7.6.4]{amrein1996}.

\begin{thm} \label{thm:LAPEnergy}
	Under the assumptions of Theorem~\ref{thm:LAP}, the following limiting absorption principle holds for every $\kappa>0$, every compact subset $I_0\subset(\kappa,\infty)$, and every $\nu>1/2$:
	\begin{align}
		\sup_{\lambda \in I_0, \mu > 0} \| E(S_\kappa) \jap{K_j}^{-\nu} (P_0-\lambda\mp\I\mu)^{-1} \jap{K_j}^{-\nu} E(S_\kappa) \| < \infty.
	\end{align}
\end{thm}

\begin{proof}
	For convenience, we choose $j=1$. Let $\kappa>0$, $\theta \in C_c^\infty((\kappa,\infty))$ a real-valued function, and set $F(P_\kappa) = \theta(P_{0,\kappa})/P_{1,\kappa}$. The operator $F(P_\kappa)$ is well-defined because if $p\in S_\kappa$ with $p_0\in \mathrm{supp}(\theta)$, then $p_1$ is separated from 0. We define the following operator on the domain $D(A_\kappa)=D(K_1)$:
	\begin{align}
		A_\kappa=\frac{1}{2}(F(P_\kappa) K_1+K_1 F(P_\kappa)).
	\end{align} 
	By Lemma~\ref{lem:ConjugateOperatorForH} below, the operator $A_\kappa$ is essentially self-adjoint, $P_{0,\kappa}\in C^\infty(\overline{A_\kappa})$, where $\overline{A_\kappa}$ is the self-adjoint closure of $A_\kappa$, and, for $0<a<b$,
	\begin{align}
		E_{0,\kappa}(\kappa+(a,b)) [P_{0,\kappa},\I \overline{A_\kappa}] E_{0,\kappa}(\kappa+(a,b)) = \theta(P_{0,\kappa}) E_{0,\kappa}(\kappa+(a,b)),
	\end{align}
	where $E_{0,\kappa}$ is the spectral measure of $P_{0,\kappa}$. If we select $\theta \in C_c^\infty((\kappa,\infty))$ such that $\theta=1$ on $\kappa+(a,b)$, then $P_{0,\kappa}$ obeys a Mourre estimate on $\kappa+(a,b)$. From Theorem~\ref{thm:LimitingAbsorptionPrinciple}, it follows that
	\begin{align}
		\sup_{\lambda \in I_0, \mu > 0} \| E(S_\kappa) \jap{\overline{A_\kappa}}^{-\nu} (P_0-\lambda\mp\I\mu)^{-1} \jap{\overline{A_\kappa}}^{-\nu} E(S_\kappa) \| < \infty.
	\end{align}
	To replace $\jap{\overline{A_\kappa}}^{-\nu}$ with $\jap{K_1}^{-\nu}$, we observe that $\jap{\overline{A_\kappa}}^{-\nu}\jap{K_1}^\nu$ is a bounded operator which commutes with $E(S_\kappa)$.
\end{proof}

\begin{prop} \label{prop:ACEnergy}
	Under the assumptions of Proposition~\ref{prop:ACLorentzCovariance}, $\mathcal{H} = Q_0\mathcal{H} \oplus \mathcal{H}_{\mathrm{ac}}(e\cdot P)$ for every time-like vector $e$.
\end{prop}

\begin{proof}
	It suffices to prove the proposition for $e=(1,0,\dots,0)$ because $e\cdot P$ is unitarily equivalent to $P_0$ if $e$ is time-like. In the proof of Theorem~\ref{thm:LAPEnergy}, we demonstrated that $P_{0,\kappa}$ obeys a Mourre estimate on $\kappa+(a,b)$ for every $0<a<b$. From Corollary~\ref{cor:AbsolutelyContinuousSpectrum}, it follows that $E(\{p\in S_\kappa \mid p_0 \in \kappa+(a,b)\})\mathcal{H} \subset \mathcal{H}_{\mathrm{ac}}(P_{0,\kappa})$. We take the union over $\kappa$, $a$, and $b$:
	\begin{align}
		E(V_+\backslash\{\mathbf{p}=0\})\mathcal{H} = \bigcup_{\kappa>0, b>a>0} E(\{p\in S_\kappa \mid p_0 \in \kappa+(a,b)\})\mathcal{H} \subset \mathcal{H}_{\mathrm{ac}}(P_0).
	\end{align}
	According to Lemma~\ref{lem:TranslationInvariantVectors}, the l.h.s.~is equal to $E(V_+\backslash\{0\})\mathcal{H}$.
\end{proof}

\begin{lem}\label{lem:ConjugateOperatorForH}
	Let $\kappa>0$, $\theta \in C_c^\infty((\kappa,\infty))$ a real-valued function, and set $F(P_\kappa)=\theta(P_{0,\kappa})/P_{1,\kappa}$. The symmetric operator
	\begin{align}
		A_\kappa=\frac{1}{2}(F(P_\kappa)K_1+K_1F(P_\kappa))
	\end{align} 
	is essentially self-adjoint on $D(K_1)$. If we denote its self-adjoint closure by $\overline{A_\kappa}$, then $P_{0,\kappa}\in C^\infty(\overline{A_\kappa})$ and 
	\begin{align}
		[P_{0,\kappa},\I \overline{A_\kappa}]=\theta(P_{0,\kappa}).
	\end{align}
\end{lem}

\begin{proof}
	\begin{enumerate}[wide, labelwidth=!, labelindent=0pt,label=(\roman*)]
		\item If $G\in C_c^\infty(\mathbb{R}^d,\mathbb{R})$ and $j\in\{1,\dots,s\}$, then $G(P_\kappa)\in C^\infty(K_j)$ and the following commutator identity holds:
		\begin{align}
			[G(P_\kappa),\I K_j] &= -\partial_{t}\big|_0 \e^{\I tK_j} G(P_\kappa) \e^{-\I tK_j} \notag \\			
			&= -\partial_{t}\big|_0 G(\cosh(t)P_{0,\kappa}-\sinh(t)P_{j,\kappa}, P_{1,\kappa}, \dots, \cosh(t)P_{j,\kappa}-\sinh(t)P_{0,\kappa}, \dots, P_{s,\kappa}) \notag \\
			&= \partial_0 G(P_\kappa) P_{j,\kappa} + \partial_j G(P_\kappa)P_{0,\kappa}. \label{eq:CommutatorLorentzBoosts}
		\end{align}
		Applying this identity to $F(P_\kappa)$, we obtain the following commutator:
		\begin{align}\label{eq:CommutatorFjKj}
			[F(P_\kappa),\I K_1] = \theta'(P_{0,\kappa}) - \theta(P_{0,\kappa})\frac{P_{0,\kappa}}{P_{1,\kappa}^2}.
		\end{align}
		The operator on the r.h.s.~is bounded. It follows that $F(P_\kappa)$ leaves $D(K_1)$ invariant. In particular, $A_\kappa$ is well-defined on $D(K_1)$.
		
		\item We apply Nelson's commutator theorem \cite[Theorem~X.36]{reed1975} to establish the essential self-adjointness of $A_\kappa$ on $D(K_1)$. Setting $N=K_1^2+1$, we define the integer scale $\mathcal{H}_k$, $k\in\mathbb{Z}$, corresponding to $N$ as the completion of $D(N^{k/2})$ with respect to the norm $\|f\|_k=\|N^{k/2}f\|$. Clearly, $A_\kappa$ is a symmetric bounded operator from $\mathcal{H}_n$ to $\mathcal{H}_{-n}$ for every $n\geq 1$. Moreover, as an operator identity from $\mathcal{H}_{n+2}$ to $\mathcal{H}_{-n-2}$,
		\begin{align}
			NA_\kappa-A_\kappa N = 2K_1 [K_1,F(P_\kappa)] K_1 + \frac{1}{2}[K_1,[K_1,[K_1,F(P_\kappa)]]].
		\end{align}
		The triple commutator in the second summand is a bounded operator on $\mathcal{H}$ according to \eqref{eq:CommutatorLorentzBoosts}. It follows that, for every $f\in \mathcal{H}_{n+2}\subset \mathcal{H}_n$, the commutator $NA_\kappa f-A_\kappa Nf$ is an element of $\mathcal{H}_{-n}$, and
		\begin{align}
			\|NA_\kappa f-A_\kappa Nf\|_{-n} \leq c \|f\|_n;
		\end{align}
		hence, the commutator $NA_\kappa-A_\kappa N$ extends to a bounded operator from $\mathcal{H}_n$ to $\mathcal{H}_{-n}$. We conclude that $A_\kappa$ is essentially self-adjoint on any core of $N$, particularly on $D(K_1^2)$. The closure of $A_\kappa$ restricted to $D(K_1^2)$ coincides with the closure of $A_\kappa$ defined on $D(K_1)$. Thus, $A_\kappa$ is essentially self-adjoint on $D(K_1)$.
		
		\item If $\chi\in C_c^\infty(\mathbb{R})$ is a function such that $\chi=1$ on $\supp(\theta)$, then $\overline{A_\kappa} = \chi(P_{0,\kappa})\overline{A_\kappa}$. In fact, if $f\in D(K_1)$, then, by \eqref{eq:CommutatorFjKj},
		\begin{align}
			\overline{A_\kappa}f = A_\kappa f = F(P_\kappa)K_1 f + \frac{1}{2}\left( \theta'(P_{0,\kappa}) - \theta(P_{0,\kappa})\frac{P_{0,\kappa}}{P_{1,\kappa}^2} \right) f = \chi(P_{0,\kappa}) A_\kappa f = \chi(P_{0,\kappa}) \overline{A_\kappa} f.
		\end{align}
		This identity extends to $f\in D(\overline{A_\kappa})$ by approximating $f$ with elements from $D(K_1)$ in the graph topology of $\overline{A_\kappa}$.
		
		\item \label{i:PisC1} We prove $P_{0,\kappa} \in C^1(\overline{A_\kappa})$. Utilising $[P_{0,\kappa},\I K_1] = P_{1,\kappa}$ (see Proposition~\ref{prop:RegularityMomentumOperators}), it is easy to verify that, for every $f\in D(K_1)\cap D(P_{0,\kappa})$,
		\begin{align}\label{eq:CommutatorP0Alambda}
			|\scp{P_{0,\kappa} f}{\overline{A_\kappa} f}-\scp{\overline{A_\kappa} f}{P_{0,\kappa} f}|
			=|\scp{f}{\theta(P_{0,\kappa})f}| \leq \|\theta(P_{0,\kappa})\|\|f\|^2.
		\end{align}
		By approximating $f\in D(\overline{A_\kappa})$ in the graph topology of $\overline{A_\kappa}$ with elements from $D(K_1)$, it follows from the previous step that \eqref{eq:CommutatorP0Alambda} is valid for $f\in D(\overline{A_\kappa})$. 		
		Moreover, $\{f\in D(\overline{A_\kappa})\mid (P_{0,\kappa}\pm\I)^{-1}f\in D(\overline{A_\kappa})\}$ contains the core $D(K_1)$. In fact, if $f\in D(K_1)$, then
		\begin{align}
			\lim_{t\to 0} \frac{1}{\I t}(\e^{\I tK_1}-1)(P_{0,\kappa}\pm\I)^{-1}f = -\I P_{1,\kappa}(P_{0,\kappa}\pm\I)^{-2}f+(P_{0,\kappa}\pm\I)^{-1}K_1f \in \mathcal{H}.
		\end{align}
		We conclude $P_{0,\kappa}\in C^1({\overline{A_\kappa}})$ by Proposition~\ref{prop:C1CommutatorCharacterisation}.
		
		\item From the previous step, it follows that the commutator $[P_{0,\kappa},\I \overline{A_\kappa}]=\theta(P_{0,\kappa})$ is a bounded operator. By similar arguments as in step \ref{i:PisC1}, we compute the higher order commutators:
		\begin{align}
			\mathrm{ad}_{-\I \overline{A_\kappa}}^k(P_{0,\kappa}) = \theta_{k-1}(P_{0,\kappa}), \ \theta_k(x) = [\theta(x)\partial_x]^k\theta(x).
		\end{align}
		All commutators are bounded operators; thus, $P_{0,\kappa}\in C^\infty(\overline{A_\kappa})$. \qedhere
	\end{enumerate}	
\end{proof}

\subsubsection{Light-cone operators}

We prove Proposition~\ref{prop:ACLorentzCovariance} for non-zero light-like vectors $e$ and Proposition~\ref{prop:ACLorentzCovarianceLightLike}. It suffices to consider $e=(1,\mp 1,0,\dots,0)$ because $e\cdot P$ is unitarily equivalent to $P_0+P_1$ or $P_0-P_1$ if $e$ is light-like.

\begin{prop}\label{prop:ACLightLike}
	Under the assumptions of Proposition~\ref{prop:ACLorentzCovariance}, $\mathcal{H} = Q_0\mathcal{H} \oplus \mathcal{H}_{\mathrm{ac}}(P_0\pm P_1)$ if $d\geq 3$ and $\mathcal{H} = E_{P_0\pm P_1}(\{0\})\mathcal{H} \oplus \mathcal{H}_{\mathrm{ac}}(P_0\pm P_1)$ if $d=2$.
\end{prop}

\begin{proof}
	From \eqref{eq:CommutationEnergyLorentz} and \eqref{eq:CommutationMomentumLorentz}, we obtain the following identity ($t,x\in\mathbb{R}$):
	\begin{align}
		\e^{\I t K_1} \e^{\I x(P_0\pm P_1)} \e^{-\I tK_1} = \e^{\I \e^{\mp t} x(P_0\pm P_1)}.
	\end{align}
	In the terminology of the example subsequent to Definition~\ref{defn:MourreEstimate}, the operator $P_0\pm P_1$ is $\pm K_1$-homogeneous. It follows that $P_0\pm P_1$ obeys a Mourre estimate on every open and bounded interval that is separated from 0. Thus, by Corollary~\ref{cor:AbsolutelyContinuousSpectrum}, $E(V_+\backslash\{p_0\pm p_1\neq 0\})\mathcal{H} \subset \mathcal{H}_{\mathrm{ac}}(P_0\pm P_1)$. If $d\geq 3$, then $E(V_+\backslash\{p_0\pm p_1\neq0\})\mathcal{H} = E(V_+\backslash\{0\})\mathcal{H}$. In fact, for $f\in\mathcal{H}$, $P_0f=\pm P_1f$ implies $P_2f=0$ by the spectrum condition. And, by Lemma~\ref{lem:TranslationInvariantVectors}, $P_2f=0$ implies $P_0f=P_1f=0$. 
\end{proof}

\begin{rem}
	Assume that an eigenstate $f\in\mathcal{H}/ Q_0\mathcal{H}$ with eigenvalue 0 of the mass operator $M=\sqrt{P_0^2-|\mathbf{P}|^2}$ exists. If $d=2$, then $P_0f=|\mathbf{P}|f=|P_1|f$ for such a massless excitation. If $f_\pm$ denotes the positive/negative momentum component of $f$, then $(P_0\mp P_1)f_\pm=0$. This illustrates that, in the case $d=2$, the subspace $E_{P_0\pm P_1}(\{0\})\mathcal{H}$ can be larger than $Q_0\mathcal{H}$.
\end{rem}

\subsection{Dilation-covariant representations} \label{ssec:DilationCovariantRepresentations}

Let $U:\mathbb{R}^d\to\mathfrak{B}(\mathcal{H})$ be a strongly continuous unitary representation of the translation group, whose generators obey the spectrum condition (Assumption~\ref{ass:SpectrumCondition}). In this subsection, we assume that $U$ is dilation-covariant, that is, a self-adjoint operator $D$ exists such that, for every $\mu\in\{0,\dots,s\}$ and $t,x\in\mathbb{R}$,
\begin{align}\label{eq:DilationCovariance}
	\e^{\I tD} \e^{\I x P_\mu} \e^{-\I tD} = \e^{\I \e^{-t}x P_\mu}.
\end{align}
We analyse the spectrum of the mass operator $M=\sqrt{P_0^2-|\mathbf{P}|^2}$, which is a well-defined self-adjoint operator due to the spectrum condition.

\begin{proof}[Proof of Proposition~\ref{prop:ACDilationCovariance}]
	In the terminology of the example subsequent to Definition~\ref{defn:MourreEstimate}, $M$ is $D$-homogeneous due to \eqref{eq:DilationCovariance}, implying $M\in C^\infty(D)$ and
	\begin{align}
		[M,\I D] = M.
	\end{align}
	Thus, $M$ obeys a Mourre estimate with conjugate operator $D$ on every open and bounded subset of $[a,\infty)$, $a>0$. The limiting absorption principle \eqref{eq:LAPMass} for the mass operator follows from Theorem~\ref{thm:LimitingAbsorptionPrinciple}. That the spectrum of $M$ is purely absolutely continuous in $(0,\infty)$ is a consequence of Corollary~\ref{cor:AbsolutelyContinuousSpectrum}. 
\end{proof}

\subsection*{Acknowledgements}

I thank Wojciech Dybalski for carefully proofreading the draft of this paper and for assistance preparing the publication. The research leading to these results received funding from National Science Centre, Poland, under the grant `Sonata Bis' 2019/34/E/ST1/00053.

\end{document}